\documentclass[11pt, reqno, a4paper]{amsart}
\usepackage[T1]{fontenc}
\usepackage[table,xcdraw]{xcolor}
\usepackage{braket}
\usepackage{graphicx}
\usepackage{tikz}
\usepackage{amssymb,amsthm}
\usepackage[margin=1in]{geometry}
\usepackage[colorlinks = true, linkcolor = blue, urlcolor  = blue, citecolor = red]{hyperref}
\usepackage{subfig}
\renewcommand{\epsilon}{\varepsilon}
\renewcommand{\phi}{\varphi}
\newcommand{\C}{\mathbb{C}}

\newcommand{\M}{\mathcal{M}}
\DeclareMathOperator{\Tr}{Tr}

\newcommand{\ketbra}[2]{|#1\rangle\langle#2|} 
\newtheorem{theorem}{Theorem}[section]
\newtheorem{definition}[theorem]{Definition}
\newtheorem{proposition}[theorem]{Proposition}
\newtheorem{corollary}[theorem]{Corollary}

\newtheorem{remark}[theorem]{Remark}

\DeclareMathAlphabet\mathbfcal{OMS}{cmsy}{b}{n}

\begin{document}

	\title{A physical noise model for quantum measurements}

\author{Faedi Loulidi}
\email{loulidi@irsamc.ups-tlse.fr}
\address{Laboratoire de Physique Th\'eorique, Universit\'e de Toulouse, CNRS, UPS, France}

\author{Ion Nechita}
\email{ion.nechita@univ-tlse3.fr}
\address{Laboratoire de Physique Th\'eorique, Universit\'e de Toulouse, CNRS, UPS, France}

\author{Clément Pellegrini}
\email{clement.pellegrini@math.univ-toulouse.fr}
\address{Institut de Mathématiques, Université de Toulouse, CNRS, UPS, France}

 \begin{abstract}
In this paper we introduce a novel noise model for quantum measurements motivated by an indirect measurement scheme with faulty preparation. Averaging over random dynamics governing the interaction between the quantum system and a probe, a natural, physical noise model emerges. We compare it to existing noise models (uniform and depolarizing) in the framework of incompatibility robustness. We observe that our model allows for larger compatibility regions for specific classes of measurements. 
 
\end{abstract}
	
	\date{\today}
	
	\maketitle
	
	\tableofcontents
	\section{Introduction}

In quantum information theory, measurements are modelled by POVMs (\emph{Positive Operator Valued Measure}). POVMs are tuples of positive semidefinite operators that sum up the identity. The POVM formalism allows one to describe the outcome probabilities of the measurement process without knowing the microscopic details of the measurement device. 
However, there exists another way of extracting the probability of obtaining some outcome without destroying the quantum system. This procedure is based on \emph{indirect measurement process}. For a given quantum state $\rho$ in a Hilbert space $\mathcal{H}_S\cong\C^d$, the indirect measurement process consists in performing a measurement of the probe on the evolved quantum state coupled to a general probe in a Hilbert space $\mathcal{H}_P\cong\C^n$. Initially, the quantum state is coupled to a probe, and the evolution of the total system is given by a unitary operator $U\in\mathcal{U}(d\,n)$. More precisely the evolution of the total system is given by: 
\begin{align*}
    U:\mathcal{H}_S\otimes\mathcal{H}_P&\to \mathcal{H}_S\otimes\mathcal{H}_P\\
    \rho\otimes\beta&\to U\big(\rho\otimes\beta\big)U^*.
\end{align*}
Now consider an observable $A$ on the probe system with spectral decomposition $A=\sum\lambda_i P_i$. In the indirect measurement process, one observes the eigenvalue $\lambda_i$ of A with probability
\begin{equation*}
    \mathbb P(i):=\Tr\Big[U(\rho\otimes \beta)U^*(I\otimes P_i)\Big].
\end{equation*}
Models of indirect measurement have been studied in the physical and mathematical community. Such models are at the cornerstone of the understanding of major experiments in quantum optics \cite{Haroche,guerlin_progressive_2007} see also \cite{carmichael,haroche2006exploring,wisemanmilburn}. For a mathematical approach where such model are precisely studied, one can consult \cite{benoist2018entropy,benoist2019invariant,PhysRevA.84.044103}. In what follows, we shall distinguish two different types of indirect measurement, a \emph{perfect measurement} and an \emph{imperfect measurement} which is interpreted as noisy measurement \cite{articleAA}.

More precisely, we say that we obtain a perfect measurement if the probe is \emph{well prepared} initially in a prescribed pure state. As an illustration, without lost of generality, consider the probe $\beta=\ketbra{0}{0}$. The probability of obtaining an outcome $\lambda_i$ is given by: 
\begin{equation*}
    \mathbb P(i)=\Tr\Big[U(\rho\otimes \ketbra{0}{0})U^*(I\otimes P_i)\Big].
\end{equation*}

An imperfect measurement can be understood physically as \emph{noise} that occurs in a quantum measurement device. To give a description of an imperfect measurement, we shall consider that,  instead of preparing the probe in the reference state $\ketbra{0}{0}$, preparation errors will result in the probe being in a state $\beta$ close to the desired state $\ketbra 0 0$. For example, $\beta$ could be a \emph{convex combination of two projectors} with a parameter $t\in[0,1]$: $\beta_t:=(1-t)\ketbra{0}{0}+t\ketbra{1}{1}$; here, $t$ plays the role of the noise parameter. In the general case, the probability of obtaining the outcome $\lambda_i$ is given by : 
\begin{equation*}
    \mathbb P(i)=\Tr\Big[U(\rho\otimes \beta)U^*(I\otimes P_i)\Big].
\end{equation*}
In both situations described above, the indirect measurement process induces an effective POVM on the system of interest $$\mathbb P(i)=\Tr [\rho\,A_i].$$
In this article, we shall describe \emph{noisy POVMs} as \emph{emergent} from a \emph{physical motivated model} based on an \emph{indirect measurement process} with an application to the \emph{compatibility of quantum measurements}.
The effect of noise has recently received attention for practical reasons, in noisy quantum computers or what is known as NISQ devices see \cite{Bharti_2022, lau2022nisq}, in quantum metrology \cite{len2022quantum}, and in the compatibility of quantum measurements \cite{designolle2019incompatibility}.

One should mention that in this work, we do not attempt a possible explanation of the measurement problem in quantum theory. However, our main focus to describe a noise model as an effective description that emerges naturally. Different noise models were introduced in the literature \cite{designolle2019incompatibility} by tacking the convex sum of the original measurement device with a trivial operator as we shall describe below. These noise models are justified by their mathematical elegance, being mixtures between the original measurement and simple, trivial measurements. In this paper, we try to go beyond the introduced noise models with a simple possible justification based on an indirect measurement process, which will allow us to model the noise effect and give a possible physical justification for the form of the noise model considered (see Fig.~\ref{fig:indirect measurement compatibility}). We will show the obtained noise model in our context is different from the one considered previously in the literature. 

Noisy POVM $\mathcal{A}^{\alpha}:=(A_1^{\alpha},\cdots,A_n^{\alpha})$ in the literature are described, by a convex combination of an initial POVM $\mathcal{A}=(A_1,\cdots,A_n)$ and a trivial POVM $\mathcal{T}:=(t_1 I,\cdots,t_n I)$ with some parameter $\alpha\in[0,1]$: 
\begin{equation}\label{eq:noisePOVM}
    A_i^{\alpha}:=\alpha\,A_i+(1-\alpha)t_i I,
\end{equation}
where we shall distinguish two main types of noise that were previously considered:
\begin{itemize}
    \item The \emph{uniform noise} where, for all $i\in[n]$ we have $t_i=1/n$. 
    \item The \emph{depolarizing noise} where, for all $i\in[n]$ we have $t_i=\Tr A_i/d$.
\end{itemize}

The different types of noise models above were introduced to study the effect of noise on the compatibility of quantum measurements. The notion of compatibility is highly explored in different contexts, such as the compatibility of quantum channels \cite{heinosaari2017incompatibility}, the compatibility dimension \cite{loulidi2021compatibility,uola2021quantum,heinosaari2021testing}, and the link between incompatibility and non-locality \cite{wolf2009measurements,loulidi2022measurement}.

In the following, we will give a brief overview of the different results in this work.  

In the case of POVMs with two outcomes, we shall assume that the system $\rho\in \C^d$ is coupled to a probe given by a two-level system $\beta\in\C^2$, that we fix initially in $\ketbra 0 0$ (in the perfect measurement context). Considering a probe observable $A=\lambda_0\ketbra 0 0+\lambda_1 \ketbra 1 1$, the probability of obtaining the outcome $i\in\{0,1\}$ system will induce a POVM with two outcomes $\mathcal{A}$. This POVM depends on the structure of the interaction between the system and the probe modelled by a bipartite unitary operator $U\in\mathcal U(d\cdot2)$. Next, to introduce our main \emph{ physical noise model} based on the \emph{indirect measurement process}, we shall assume that initially the probe is given by a general two-level state $\beta\in\C^2$ which is interpreted as a noisy perturbation of $\ketbra 00$ (the probe is not perfectly prepared or is perturbed before interacting with the system). In this context we still consider the same probe observable $A=\lambda_0\ketbra 0 0+\lambda_1 \ketbra 1 1$. This time the probability of obtaining the outcome $i\in\{0,1\}$ will induce the POVM $\mathcal{A}^{\beta}$ that will depend on $\beta$ and the block elements of $U$ (as in the perfect case). In our model, we shall assume that the unitaries generating the evolution are random. This consideration is an approximation of a possible description of different unknown degrees of freedom generating the evolution with a random Hamiltonian of the system coupled with a probe. The use of random unitaries will allow us to simplify the description of our model as a first approximation. Our noise model \emph{emerges naturally}, by considering a random model for $U$ and taking the average of $\mathcal{A}^{\beta}$ with respect to the distribution of the interaction unitary $U$. This procedure gives rise to the \emph{noisy two outcome POVM} $\mathbb E[\mathcal{A}^{\beta}]=(\mathbb{E}[A_0^{\beta}],\mathbb{E}[A_1^{\beta}])$, where its elements are given by: 
\begin{equation*}
    \mathbb{E}[A_0^{\beta}]=\beta_{00}\,A_0+\beta_{11}\,\frac{\Tr[A_1]}{d}I_d\quad\text{and}\quad\mathbb{E}[A_1^{\beta}]=\beta_{00}\,A_1+\beta_{11}\,\frac{\Tr[A_0]}{d}I_d.
\end{equation*}

Naturally, we obtain a noisy POVM of the form \eqref{eq:noisePOVM}. The particularity is that it is close to the depolarizing noise model where the indices of the effect operators inside the trace are switched. While usual models such as uniform or depolarizing are usually introduced as ad-hoc noisy POVMs, our model is physically motivated through indirect measurement. Interestingly, we end up with a "model close to" the depolarising model where the usual noise parameters are switched.

In the rest of the article, we give a complete generalization of our noise model for any number of outcomes $i\in[0,N]$ for arbitrary $N\in\mathbb N$. In this general context, we describe physically motivated effective POVM $\mathcal{A}^{\beta}$ with $N+1$ outcomes. Again, we concentrate on averaged POVM $\mathbb E[\mathcal{A}^{\beta}]$ by considering random models for the involved unitary operator.

Finally, we shall exploit our noise model by providing some applications for the compatibility of quantum measurements. In particular, we compare it with the usual noise models considered in the literature. 
\bigskip

The paper is organized as follows. In Section \ref{sec:noisy-POVMs}, we recall the notion of compatibility and the different type of noise models known in the literature. In Section \ref{sec: two level model}, we introduce our main \emph{ physically noise model based on an indirect measurement process} for POVMs with two outcomes.
In Section \ref{sec: noise model in higher dim}, we will generalize the results for POVMs with more than two outcomes. In Section \ref{sec: examples}, we will give examples of the application of our noise model to the compatibility of quantum measurements. 
 
\section{Quantum measurements}\label{sec:noisy-POVMs}
In this section, we will recall basic notions from quantum information theory. We shall recall the (in-)compatibility of quantum measurement and the different types of noise models established in the literature.
In Quantum Information theory, a quantum state is described by the framework of \emph{density matrices} in a finite-dimensional Hilbert space $\mathcal{H}\cong\C^d$. Formally we shall denote by the set of quantum states by $\mathcal{M}_d^{1,+}$ defined as \begin{equation*}
    \M^{1,+}_d := \{ \rho \in \M_d \, : \, \rho \geq 0 \text{ and } \Tr \rho = 1\},
\end{equation*}
The set of density matrices encodes all the information about the physical system. 
One of the main differences between classical mechanics and quantum theory, the measurement outcomes of a given experiment are intrinsically probabilistic. The celebrated \emph{Born rule} allows us to obtain the probability of a given outcome. In general, the measurement process is characterized by Positive Operator Valued Measures (POVMs)\cite{nielsen00}. Formally a \emph{positive operator valued measure} on $\M_d$ with $N$ outcomes is a $N$-tuple $\mathcal A=(A_1, \ldots, A_N)$ of self-adjoint operators from $\M_d$ which are positive semidefinite and sum up to the identity:
    $$\forall i \in [N], \quad A_i \geq 0 \qquad \text{ and } \qquad \sum_{i=1}^N A_i = I_d.$$
    When measuring a quantum state $\rho$ with the apparatus described by $\mathcal A$, we obtain a random outcome from the set $[N]$:
    $$\forall i \in [N], \qquad \mathbb P(\text{outcome} = i) = \Tr[\rho A_i].$$
 We shall write $[0,N]:= \{0,1, \ldots, N\}$ for the set of the $N+1$ positive integers and the set for $N$ positive elements are given by $[N]:=\{1,\cdots,N\}$. In particular, if the POVM framework reduces to the \emph{projective measurements} (or \emph{von Neumann measurements}) the measurement are projectors of the form $A_i=\ketbra{a_i}{a_i}$. 
 The POVM framework, allows us to understand only the measurement outcome which is relevant from given experiments without knowing the microscopic details of the measurement device see figure \ref{fig:measurement}
for an illustration.
Another main difference between classical physics and quantum mechanics is the notion of \emph{compatibility}. In classical mechanics, the measurement does not affect the physical system. However, it is well known since the discovery of the quantum theory the existence of measurements that cannot be measured at the same time we called incompatible, and others that are compatible. We say that two POVMs $\mathcal{A}=(A_1,\cdots,A_N)$ and $\mathcal{B}=(B_1,\cdots,B_M)$ are compatible if there exists a third POVM $\mathcal{C}=(C_{11},\cdots,C_{NM})$ from where they are marginals, see figure \ref{fig:compatibility} for an illustration. Formally, two POVMs $\mathcal{A}$ and $\mathcal{B}$ are compatible if there exists a joint POVM $\mathcal{C}=(C_{11},\cdots,C_{NM})$ such that:
    \begin{equation*}
        \forall i \in [N],\,  A_i= \sum_{j=1}^M C_{ij},\quad
        \forall j \in [M],\,  B_j = \sum_{i=1}^NC_{ij}.
    \end{equation*}
    In particular, the definition of compatibility reduces to the commutativity for Projective measurements. Moreover, the definition of compatibility can be extended naturally to \emph{tuples of POVMs} see \cite{heinosaari2016invitation} and the references therein for more details. 
    \begin{figure}
\centering
    \includegraphics[height=0.7\textheight]{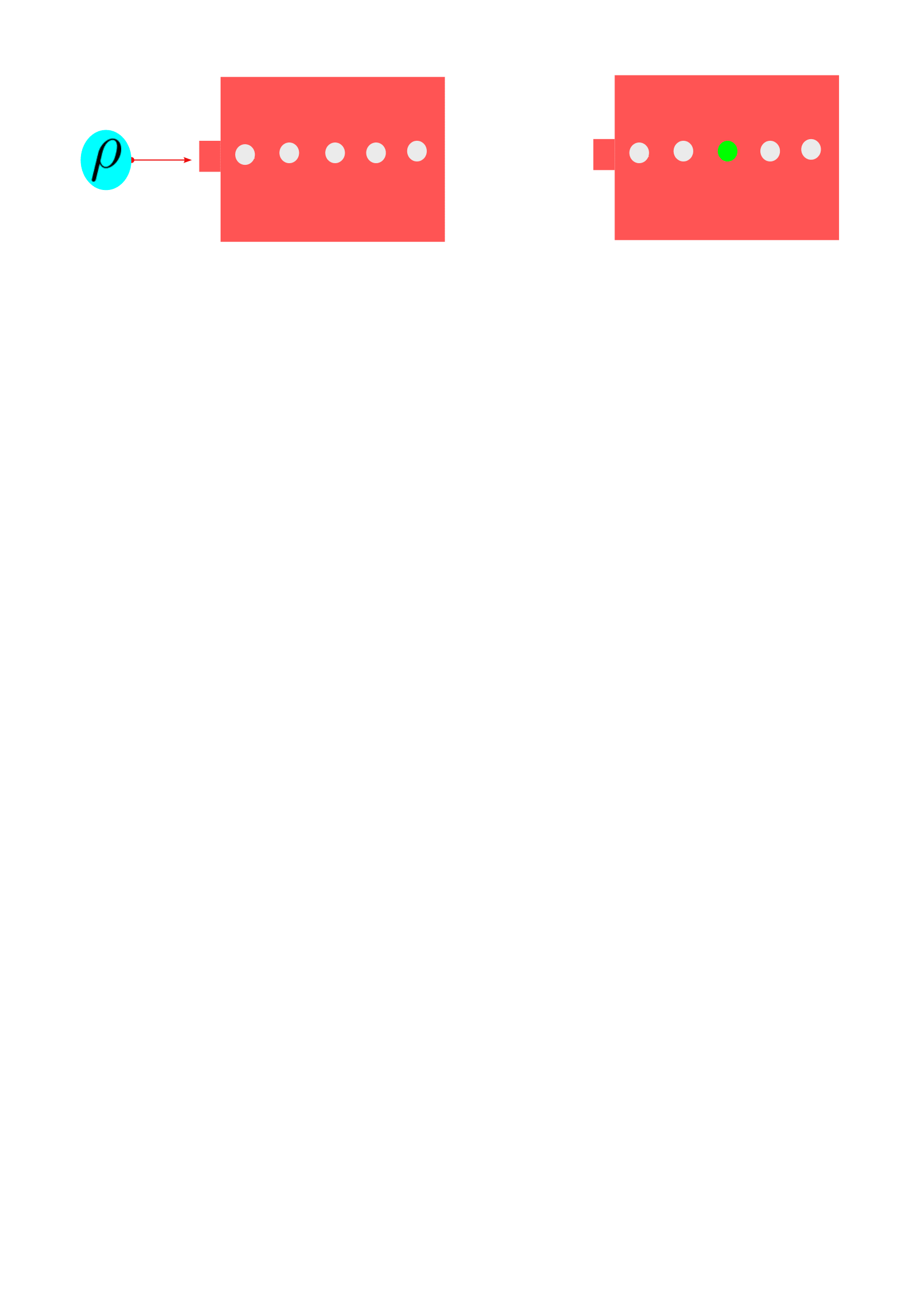}
    \vspace{-14cm}
    \caption{Diagrammatic representation of a quantum measurement apparatus. The device has an input canal and a set of $5$ LEDs which will turn on when the corresponding outcome is achieved. After the measurement is performed, the particle is destroyed, and the apparatus displays the classical outcome (here, $3$).}\label{fig:measurement}
    \end{figure}
    In the following, we shall recall different types of \emph{noise models} considered in the literature. A noise model is a POVM constructed from a convex combination with a parameter $\alpha\in[0,1]$ of an original POVM $\mathcal{A}$ and a \emph{trivial} POVM $\mathcal{T}$. We shall focus on two different noise models well studied in the literature: the uniform noise, and the depolarizing noise.
    
    \begin{figure}[htb!]
    \centering
    \includegraphics[width=0.9\textwidth]{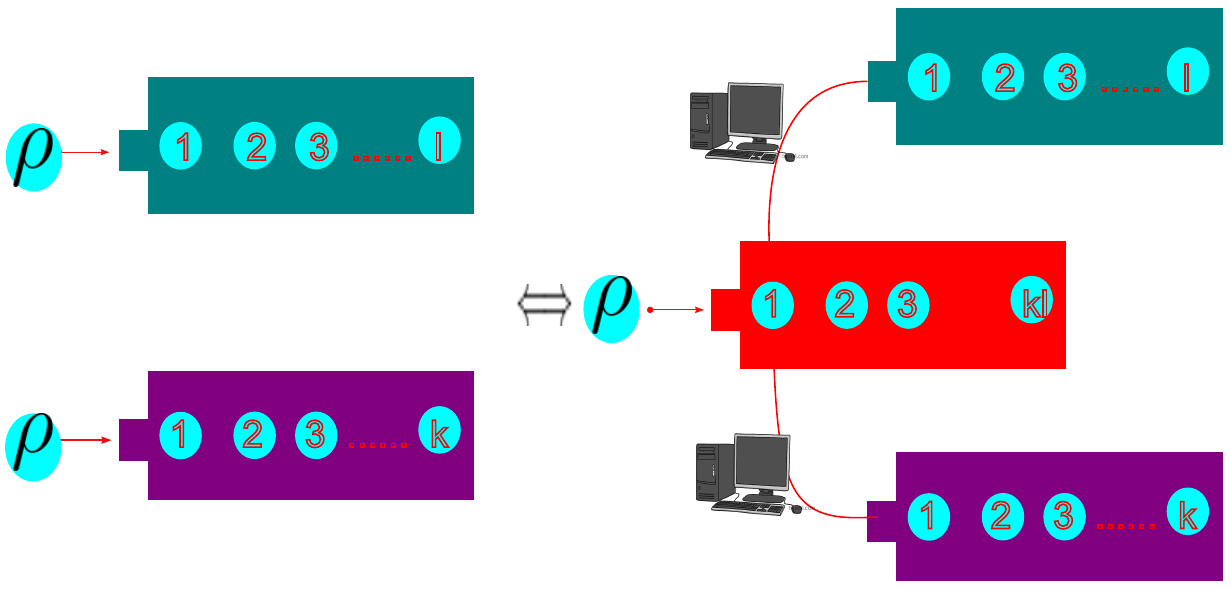}
    \vspace{-11cm}
    \caption{The joint measurement of $A$ and $B$ is simulated by by a third measurement $C$, followed by classical post-processing.}
    \label{fig:compatibility}
\end{figure}
 \begin{itemize}
     \item The uniform noise model for $\alpha\in[0,1]$ is defined by: 
     \begin{equation*}         \mathcal{A}\to\mathcal{A}^{\alpha}:=\alpha \mathcal{A}+(1-\alpha)\mathcal{T}, 
     \end{equation*}
     where $\mathcal{T}:=(t_1 I,\cdots,t_N I)$, with $t_i:=1/N$ for all $i\in[N]$.
     \vspace{5mm}
     \item The depolarizing noise for $\alpha\in[0,1]$ is defined by:
     \begin{equation*}         \mathcal{A}\to\mathcal{A}^{\alpha}:=\alpha \mathcal{A}+(1-\alpha)\mathcal{T}, 
     \end{equation*}
     where $\mathcal{T}:=(t_1 I,\cdots,t_N I)$, with $t_i:=\Tr A_i/d$ for all $i\in[N]$.
     \end{itemize}
     Noise models play several important roles in quantum information and more specifically in the compatibility of quantum measurements. One should only mention that for a given incompatible measurement, one can ask how much noise one can add to the original POVMs in order to make them compatible. This question will also be addressed when we will provide applications of our physical noise model to the compatibility question and compare it with the standard noise model in the last section of this paper.
\section{Effective noise model via indirect measurements: the two outcomes case.}\label{sec: two level model}
In this section, we will describe our \emph{effective noise model} obtained from an \emph{indirect measurement process} in the case of POVMs with two outcomes. As outlined in the Introduction, POVMs can be induced by indirect measurement (see figure \ref{fig:indirect measurement compatibility} for a pictorial representation). Here we shall consider a probe of dimension $2$ and an observable with two outcomes. This will indeed give rise to two POVMS with two outcomes. Then introducing noise in the probe preparation will induce noisy POVMs that we can compare with the usual noise models in the literature.

In the sequel, we shall describe the induced POVMs in terms of the unitary interaction between the system and the probe. Briefly speaking the description of the interaction will impose the form of the POVMs. In the case of dimension $2$ (for the probe), we are able to describe an equivalence between the form of the induced POVMs and the structure of the unitary interaction. This structure combined with noise in the probe allows us to introduce random models which allow us to derive effective POVMs that enter into the form \eqref{eq:noisePOVM}. Up to our knowledge, the main difference is that naturally, we obtain noises that have never been studied in the literature and which are physically motivated. 

Subsection \ref{subsection: the induced POVM} concerns the model of indirect measurement in the perfect situation and 
 Subsection \ref{subsec: General two level probe} concerns the models where we introduce noise in the probe preparation.

\begin{figure}[htb!]
    \centering
    \includegraphics[width=1\textwidth]{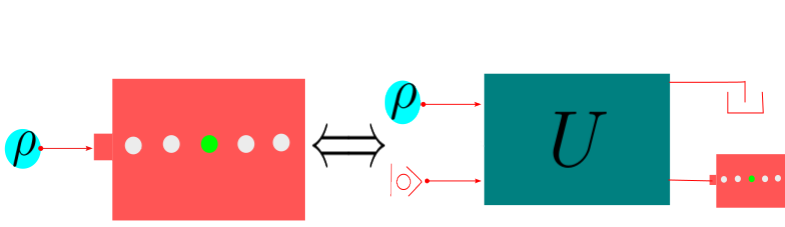}
    \caption{The measurement of a quantum state on a POVM is equivalent to an indirect measurement process, where one measures the probe on an evolved quantum state coupled to a prepared probe on $\ketbra{0}{0}$. By destroying the probe in this process, we recover the initial POVM.}
    \label{fig:indirect measurement compatibility}
\end{figure}
\subsection{The two outcomes POVM induced by a perfect indirect measurement scheme}\label{subsection: the induced POVM}
Here we consider $\mathcal H_P=\mathbb C^2$. In the following, we will assume that initially the probe is perfectly prepared in the state $\ketbra{0}{0}$.  The evolution of the initial state $\rho$ of the system coupled to the prepared probe on $\ketbra{0}{0}$, where we have: 
\begin{equation}
\rho\otimes\ketbra{0}{0}\to U\Big(\rho\otimes\ketbra{0}{0}\Big)\,U^*,
\end{equation}
where $U\in\mathcal{U}(2d)\subseteq \mathcal{M}_d\otimes\mathcal{M}_2$ given by: 
\begin{equation}
U=\sum_{i,j=0}^1 U_{ij}\otimes\ketbra{i}{j}=\begin{bmatrix}
U_{00} & U_{01}\\
U_{10} & U_{11}\end{bmatrix},
\end{equation}
with $U_{ij}\in\mathcal{M}_d(\mathbb{C})$ are the blocks of $U$.

We consider an observale $A=\lambda_0\ketbra 00+\lambda_1\ketbra 11$ on $\mathcal H_P$. We denote by $\mathbb{P}(i)$ the probability of obtaining the outcome $i\in\{0,1\}$ by measuring the probe in the evolved total system, which is given by:
\begin{equation}
\mathbb{P}(i):=\Tr\Big[U\,\rho\otimes\ketbra{0}{0}\,U^*\Big(I_2\otimes\ketbra{i}{i}\Big)\Big].
\end{equation}
In the following proposition, we will show the probe measurement on the total evolved system will induce an effective POVM $\mathcal{A}$.
\begin{proposition}\label{prop: effective POVM}
The probability of obtaining the outcome $i\in\{0,1\}$ induces a POVM with two outcomes $\mathcal{A}=(A_0,A_1)$ given by:
\begin{equation}
\mathcal A:=(A_0,A_1)\quad \text{with} \quad A_0:=U^*_{00}\,U_{00} \quad \text{and}\quad A_1:=U^*_{10}\,U_{10}.
\end{equation}
\end{proposition}
\begin{proof}
Let $\mathbb P(0)$ and $\mathbb P(1)$ the probability of obtaining the outcome $0$ and $1$, an explicit computation shows that the probability of obtaining the outcome $i=0$ is :
\begin{align*}
\mathbb{P}(0)&=\Tr\Big[U\,\Big(\rho\otimes\ketbra{0}{0}\Big)\,U^*\,\Big(I_d\otimes\ketbra{0}{0}\Big)\Big]\\
	         &=\Tr\Big[\sum_{i,j=0}^1\,U_{ij}\otimes\ketbra{i}{j}\Big(\rho\otimes\ketbra{0}{0}\Big)\sum_{a,b=0}^1 U^*_{ab}\otimes\ketbra{b}{a}\Big(I_d\otimes\ketbra{0}{0}\Big)\Big]\\  
	         &=\sum_{i,j=0}^1\,\sum_{a,b=0}^1\Tr\Big[U_{ij}\,\rho \,U^*_{ab}\otimes\ket{i}\braket{j|0}\braket{0|b}\braket{a|0}\bra{0}\Big]\\
	    &=\Tr\Big[U^*_{00}\,U_{00}\,\rho\Big]=\Tr\Big[A_0\rho\Big],
        \end{align*}
	  where we have defined $A_0:=U^*_{00}U_{00}$.
   
   A simple computation of $\mathbb P(1)$ as the one before, shows that: \begin{equation*}
       \mathbb P(1)=\Tr\Big[U^*_{10}\,U_{10}\,\rho\Big]=\Tr\Big[A_1\rho\Big],
   \end{equation*}
	   
	   where we have defined $A_1:=U^*_{10}U_{10}.$
    
Using that $U\in\mathcal{U}(2d)$ is a unitary matrix, one can check easily that
$
U^*_{00}\,U_{00}\,+\,U^*_{10}\,U_{10}=I_d$
which implies $A_0+A_1=I_d$ which is the requirement for $A:=(A_0,A_1)$ being a POVM.

\end{proof}
In the proposition above, we have seen that every unitary $U\in\mathcal{U}(2d)$, induces by an indirect measurement scheme, an effective POVM $\mathcal{A}=(A_0,A_1)$
on the target system $\C^d$. We now ask the reverse question: \emph{given a fixed POVM $\mathcal{A}$ on $\C^d$,what is the set of interaction unitary operators $U$ which yield $\mathcal{A}$ as an effective POVM?} 
This question motivates the following definition, where we will fix the POVM elements and search for unitaries where their first column is given by the POVM elements.
\begin{definition}\label{def: the unitary set}
Given a fixed two outcomes POVM $\mathcal{A}=(A_0,A_1)$, we define a subset of unitary matrices $\mathbb{U}(A_0,A_1)$ as:
\begin{equation}
\mathbb{U}(A_0,A_1):=\Big\{U\in\mathcal{U}(2d)\Big|\,A_0=U^*_{00}\,U_{00},\,A_1=U^*_{10}\,U_{10}\Big\}\subseteq \mathcal{U}(2d).
\end{equation}
\end{definition}

In the following proposition, we will show that the set $\mathbb{U}(A_0,A_1)$ can be completely characterized with the help of three unitary matrices $V,W,Z\in\mathcal{U}(d).$
\begin{proposition}\label{prop: unitray set}
The set $\mathbb{U}(A_0,A_1)$ is completely characterized as follows:

\begin{align*}
\mathbb{U}(A_0,A_1)=\Big\{V,W,Z\in\mathcal{U}(d)&\Big|\, U_{00}=V\,\sqrt{A_0},U_{10}=W\,\sqrt{A_1},\\
&U_{01}= V\,\sqrt{A_1}\,Z^*\,,\,U_{11}=- W \,\sqrt{A_0}\,Z^*\Big\}.
\end{align*}
In other words $U\in\mathbb{U}(A_0,A_1)$ if and only if there exists $V,W,Z\in\mathcal{U}(d)$ such that
\begin{equation}
U=\begin{bmatrix}
V\sqrt{A_0}& V\,\sqrt{A_1}\,Z^*\\
W\sqrt{A_1} & -W\,\sqrt{A_0}\,Z^*\end{bmatrix}.
\end{equation}
\end{proposition}
\begin{proof}
 Let $\mathbb{W}(A_0, A_1)$ be the set from the right-hand side in the statement. Straightforward matrix block computations easily yield the first inclusion 
 \begin{equation}\mathbb{W}(A_0, A_1)\subseteq \mathbb{U}(A_0, A_1)\end{equation}

Let us concentrate on the reverse inclusion:
\begin{equation}
\mathbb{U}(A_0,A_1)\subseteq\mathbb{W}(A_0,A_1).
\end{equation}
To this end let $U\in\mathbb{U}(A_0,A_1)$, that is

\begin{equation}\label{eq: A_0 and A_1}
A_0=U^*_{00}U_{00}\quad \text{and}\quad A_1=U^*_{10}U_{10}.
\end{equation}

Let the polar decomposition of $U_{00}$ given by:  
\begin{equation}\label{eq: polar decomp U_00}
U_{00}=V\,P_0,
\end{equation}
where $V\in\mathcal{U}(d)$ and $P_0$ a positive semidefinte matrix.

By combining $A_0$ from equations \eqref{eq: A_0 and A_1} and \eqref{eq: polar decomp U_00}: 
\begin{equation}
A_0=U^*_{00}U_{00}=(V\,P_0)^*V\,P_0=P_0^2\implies P_0=\sqrt{A_0},
\end{equation}
we deduce that $U_{00}$ is given by:
\begin{equation}
U_{00}=V\sqrt{A_0}.
\end{equation}

Similarly, let the polar decomposition of $U_{10}$ given by: 
\begin{equation}
U_{10}=W\,P_1,
\end{equation}
where $W\in\mathcal{U}(d)$ and $P_1$ a positive semidefinite matrix. 

We have that: \begin{equation}
    A_1=U^*_{10}U_{10}=(W\,P_1)^*W\,P_1=P_1^2\implies P_1=\sqrt{A_1},
    \end{equation}
where we can deduce the form of $U_{10}$, which is given by:
\begin{equation}\label{eq: U_10}
U_{10}=W\sqrt{A_1}.
\end{equation}

From the unitary property of $U\in\mathcal{U}(2d)$, we shall determine the other blocks of $U$.
\begin{equation}
U^*U=I_{2d}\implies(\text{S}_1):\begin{cases}
U^*_{00}\,U_{00}\,+\,U^*_{10}\,U_{10}=I_d&\\
U^*_{00}\,U_{01}\,+\,U^*_{10}\,U_{11}=0&\\
U^*_{01}\,U_{00}\,+\,U^*_{11}\,U_{10}=0&\\
U^*_{01}\,U_{01}\,+\,U^*_{11}\,U_{11}=I_d&
\end{cases}
\end{equation}
and \begin{equation}\label{eq: S_2}
UU^*=I_{2d}\implies(\text{S}_2):\begin{cases}
U_{00}\,U^*_{00}\,+\,U_{01}\,U^*_{01}=I_d&\\
U_{00}\,U^*_{10}\,+\,U_{01}\,U^*_{11}=0&\\
U_{10}\,U^*_{00}\,+\,U_{11}\,U^*_{01}=0&\\
U_{10}\,U^*_{10}\,+\,U_{11}\,U^*_{11}=I_d&
\end{cases}
\end{equation}
By using the fourth equation from the system $(\text{S}_2)$ \eqref{eq: S_2} and the equation \eqref{eq: U_10}: 
\begin{equation}
WA_1W^*\,+\,U_{11}\,U^*_{11}=I_d,
\end{equation}
where from the fact that $A_0+A_1=I_d$, we deduce that:
\begin{equation}\label{eq: A_0}
(U_{11}^*\,W)^*\,U_{11}^*\,W=A_0.
\end{equation}
Let the polar decomposition of $U_{11}^*\,W$ given by: 
\begin{equation}\label{eq: U_11 polar decomp}
U_{11}^*W=-Z\,Q,\end{equation}
where $Z\in\mathcal{U}(d)$ and positive semidefinite $Q$.

By combining the equations \eqref{eq: A_0} and \eqref{eq: U_11 polar decomp}, we have:
\begin{equation}
(U_{11}^*\,W)^*\,U_{11}^*\,W=Q^2=A_0\implies Q=\sqrt{A_0},
\end{equation}
where we can deduce the form of $U_{11}$, which is given by : \begin{equation}
U_{11}=-W\,\sqrt{A_0}\,Z^*.
\end{equation}

The only remaining element we need to determine is $U_{01}$, for that we shall use the first equation of the system $(\text{S}_2)$ \eqref{eq: S_2} and from $A_0+A_1=I_d$, we have: 
\begin{equation}\label{eq: A_1}
U_{00}\,U^*_{00}\,+\,U_{01}\,U^*_{01}=I_d\iff (U_{01}^*\,V)^*U_{01}^*\,V=A_1.
\end{equation}
Let the polar decomposition of $U_{01}^*\,V$ given by:
\begin{equation}\label{eq: U_01 polar decomp}
U_{01}^*\,V=\Gamma\,\tilde{Q},
\end{equation}
with $\Gamma\in\mathcal{U}(d)$ and $\tilde{Q}$ positive matrix. 

By combining the equations \eqref{eq: U_01 polar decomp} and \eqref{eq: A_1}, we have:
\begin{equation}
(U_{01}^*\,V)^*U_{01}^*\,V=\tilde{Q}^2=A_1\implies\tilde{Q}=\sqrt{A_1},
\end{equation}
where we can deduce the form of $U_{01}$, which is given by: 
\begin{equation}
     U_{01}=V\,\sqrt{A_1}\Gamma^*.
\end{equation}
The only remaining matrix to determine is $\Gamma$, for that we shall use the second equation of $(\text{S}_2)$ \eqref{eq: S_2} and analyze the invertibility of the effect operators $A_0$ and $A_1$. 

From the second equation of $(\text{S}_2)$ \eqref{eq: S_2} we have:
\begin{equation}
U_{00}\,U^*_{10}\,+\,U_{01}\,U^*_{11}=0\iff\sqrt{A_1}\,\sqrt{A_0}=\sqrt{A_1}\,\Gamma^*\,Z\,\sqrt{A_0}
\end{equation}

Remark that $\mathrm{Ran}(A_0)+\mathrm{Ran}(A_1)=\C^d$, to discuss the inversibility $A_0$ and $A_1$, we shall distinguish two cases:
\begin{itemize}
    \item Assuming $A_0$ invertible and $A_1$ invertible only on its range:
    
    Let $\C^d\ni x=x_0+x_1$ with $x_i\in \mathrm{Ran}(A_i)$, and $y_1^T:=x_1^T(\sqrt{A_1})^{-1}$, we have: 
\begin{equation}
y_1^T\sqrt{A_1}=x_1^T\,I_d\,|_{\mathrm{Ran}(A_1)}=x_1^T\,I_d|_{\mathrm{Ran}(A_1)}\Gamma^*\,Z.
    \end{equation}
    where $I_d|_{\mathrm{Ran}(A_1)}$ is the identity on the range of $A_1$. 
    
    Hence, we have that: 
    \begin{equation}
    \Gamma^*Z\,|_{\mathrm{Ran}(A_1)}=I_d\,|_{\mathrm{Ran}(A_1)}\implies \Gamma|_{\mathrm{Ran}(A_1)}=Z|_{\mathrm{Ran}(A_1)}
    \end{equation}
    \item Assuming $A_1$ invertible and $A_0$ invertible only on its range:
    
    As before $\C^d\ni x=x_0+x_1$ with $x_i\in \mathrm{Ran}(A_i)$, and $y_0:=(\sqrt{A_0})^{-1}x_0$ we have: \begin{equation}
    \sqrt{A_0}\,y_0=\,I_d\,|_{\mathrm{Ran}(A_0)}x_0=\Gamma^*\,Z\,I_d|_{\mathrm{Ran}(A_1)}\,x_0,
    \end{equation}
    where $I_d|_{\mathrm{Ran}(A_0)}$ is the identity on the range of $A_0$. 
    
    We obtain that: 
    \begin{equation}
    \Gamma^*Z\,|_{\mathrm{Ran}(A_0)}=I_d\,|_{\mathrm{Ran}(A_0)}\implies \Gamma|_{\mathrm{Ran}(A_0)}=Z|_{\mathrm{Ran}(A_0)}.
    \end{equation}
By combining the two consequences above we obtain:
\begin{equation}
\Gamma=\Gamma|_{\mathrm{Ran}(A_0)}+\Gamma|_{\mathrm{Ran}(A_1)}=Z|_{\mathrm{Ran}(A_0)}+Z|_{\mathrm{Ran}(A_0)}=Z
\end{equation}
\end{itemize}
Therefore, we can deduce in both cases where $A_0$ and $A_1$ are both invertible and only one of the two is invertible, that $U_{01}$ is given by:
\begin{equation}
U_{01}=V\,\sqrt{A_1}\,Z^*.
\end{equation}
This ends the proof of the second inclusion, therefore the end of the proof of the proposition.

\end{proof}

\begin{remark}
The set $\mathbb U(A_0,A_1)$ is not a subgroup of $\mathcal{U}(2d)$, one need to check that $I\notin\mathbb U(A_0,A_1)$. From the definition given above we have $U_{00}=V\sqrt{A_0}$ and $U_{10}=W\sqrt{A_1}$ for $V,W\in\mathcal{U}(d)$. Let $I_{2d}\in\mathcal{U}(2d)$ we have $U_{00}=V\,\sqrt{A_0}=I_d$ and $U_{10}= W\,\sqrt{A_1}=0_d$, hence one can check that $A_0+A_1\neq I_d$, therefore $\mathbb U(A_0,A_1)$ is not a subgroup of $\mathcal{U}(2d)$.
\end{remark}

\begin{remark}\label{remark: Group orbit}
Let the action $\circ':\,\mathcal{U}(d)^{\times 3}\curvearrowright \mathbb{U}(A_0,A_1)$ defined by:
\begin{align*}
\circ':\,\mathcal{U}(d)^{\times 3}\times \mathbb{U}(A_0,A_1)&\to \mathbb{U}(A_0,A_1),\\
\Big((V,W,Z),U\Big)&\to(V,W,Z)\circ'U.
\end{align*}
For $U\in\mathbb U(A_0,A_1)$, the action of $\circ'$ is defined by:
\begin{align*}
(V,W,Z)\circ U:&=
\begin{bmatrix}
V&0\\
0 &W\end{bmatrix}\begin{bmatrix}
V'\sqrt{A_0}& V'\,\sqrt{A_1}\,Z'^*\\
W'\sqrt{A_1} & -W'\,\sqrt{A_0}\,Z'^*\end{bmatrix}\begin{bmatrix}
I& 0\\
0 & Z^*\end{bmatrix}\\
&=\begin{bmatrix}
\tilde V\sqrt{A_0}& \tilde V\,\sqrt{A_1}\,\tilde Z'^*\\
\tilde W'\sqrt{A_1} & -\tilde W'\,\sqrt{A_0}\,\tilde Z'^*\end{bmatrix}\in\mathbb{U}(A_0,A_1),
\end{align*}
where $\tilde V=VV'$,$\tilde W=WW'$ and $\tilde Z=ZZ'$ are in $\mathcal{U}(d)$.

With this at hand, one can understand the set $\mathbb{U}(A_0,A_1)$ in the proposition \ref{prop: unitray set} as the orbit of the simple matrix $\tilde U$ given by:
$$\tilde U:=\begin{bmatrix}
\sqrt{A_0}& \sqrt{A_1}\\
\sqrt{A_1} & -\sqrt{A_0}\end{bmatrix}.$$
We can also check that the following property holds:
\begin{align*}
\Big((V,W,Z)\circ '(V',W',Z')\Big)\circ \tilde U=(VV',WW',ZZ')\circ\tilde{U}.
\end{align*}
where $\circ'$ is the group action on $\mathcal{U}(d)^{\times3}$.
\end{remark}

\subsection{The noisy two outcomes POVM induced by an imperfect indirect measurement scheme}\label{subsec: General two level probe}
To introduce our main \emph{physical noise model}, we shall assume that initially the quantum state $\rho$ is coupled to a general two-level probe $\beta$, and the unitaries are random. The imperfect measurement of the probe will induce a general POVM $\mathcal{A}^{\beta}$. The average over the randomness of the POVM $\mathcal{A}^{\beta}$ will induce our physical noise model.

In the basis $\{\ket 0,\ket 1\}$ we write $\beta\in\mathcal{M}_2^{1,+}$ as: 
\begin{equation*}
\beta=\sum_{i,j=0}^1\beta_{ij}\ketbra{i}{j}\in\mathcal{M}_{2}^{1,+}.
\end{equation*}
Again, we consider the evolution of the form
\begin{equation*}
    \rho\otimes\beta\to U\Big(\rho\otimes\beta\Big)U^*,
\end{equation*}
with $U\in\mathbb U(A_0,A_1)$. 

We shall denote by $\mathbb P_{\beta}(i)$ the probability of obtaining the outcome $i\in\{0,1\}$ on the evolved quantum state $\rho$ coupled to the probe $\beta$. In the following Proposition, we shall show that the outcome probability of obtaining $i\in\{0,1\}$ will induce an effective POVM $\mathcal{A}^{\beta}$ that will depend on the probe.
\begin{proposition}\label{Prop: general two level probe}
Let the evolution of two-level probe $\beta$ coupled to a quantum state $\rho$ governed by $U\in\mathbb{U}(A_0,A_1)$ with associated unitary operators $V,W,Z$. The probability of obtaining the outcome $i\in\{0,1\}$ given by $\mathbb P_{\beta}(i)$ induces an effective POVM with two outcomes $\mathcal{A}^{\beta}$ given by:
\begin{equation*}
    \mathcal{A}^{\beta}=(A_0^{\beta},A_1^{\beta}),
\end{equation*}
where $A_0^{\beta}$ and $A_1^{\beta}$ are explicitly given by: 
\begin{align*}
    A_0^{\beta}&=\beta_{00}\,A_0+\beta_{01}\,Z\,\sqrt{A_1}\,\sqrt{A_0}+\beta_{10}\,\sqrt{A_1}\,\sqrt{A_0}\,Z^*+\beta_{11}\,Z^*\,A_1\,Z,\\
    A_1^{\beta}&=\beta_{00}\,A_1-\beta_{01}\,Z\,\sqrt{A_1}\,\sqrt{A_0}-\beta_{10}\,\sqrt{A_1}\,\sqrt{A_0}\,Z^*+\beta_{11}\,Z^*\,A_0\,Z. 
\end{align*}
\end{proposition}

\begin{proof}
An explicit computation of $\mathbb{P}_{\beta}(i)$ shows that:
\begin{align*}
    \mathbb P_{\beta}(i)&:=\Tr\Big[U\,(\rho\otimes\beta)\,U^*\,(I_d\otimes\ketbra{i}{i})\Big]\\
    &=\sum_{c,d=0}^1\beta_{cd}\,\Tr\Big[U^*_{id}\,U_{ic}\,\rho\Big]\\
    &=\Tr[A_i^{\beta}\rho],
\end{align*}
where we have defined $A_i^{\beta}:=\sum_{c,d=0}^1\beta_{cd}\,U^*_{id}\,U_{ic}.$

By using Proposition \ref{prop: unitray set}, the effects of the POVM $\mathcal{A}^{\beta}$ are given explicitly by:
    \begin{align*}
    &A_0^{\beta}:=\sum_{c,d=0}^1\beta_{cd}\,U^*_{0d}\,U_{0c}=\beta_{00}\,A_0+\beta_{01}\,Z\,\sqrt{A_1}\,\sqrt{A_0}+\beta_{10}\,\sqrt{A_1}\,\sqrt{A_0}\,Z^*+\beta_{11}\,Z^*\,A_1\,Z\\
    &A_1^{\beta}:=\sum_{c,d=0}^1\beta_{cd}\,U^*_{1d}\,U_{1c}=\beta_{00}\,A_1-\beta_{01}\,Z\,\sqrt{A_1}\,\sqrt{A_0}-\beta_{10}\,\sqrt{A_1}\,\sqrt{A_0}\,Z^*+\beta_{11}\,Z^*\,A_0\,Z
    \end{align*}
\end{proof}

To introduce our noise model from an indirect measurement process, we shall define a \emph{centered 1-design} probability measure on $\mathcal{U}(d)$.
\begin{definition}\label{def: 1 design}
A probability distribution $\mu$ on $\mathcal U(d)$ is called a \emph{centered 1-design} if
$\forall X\in\M_d$

\begin{equation}
\mathbb{E}_{Z\sim\mu}[Z\,X]=\mathbb{E}_{Z\sim\mu}[Z^*\,X]:=0
\end{equation}
and 
\begin{equation}
\mathbb{E}_{Z\sim\mu}[Z\,X\,Z^*]:=\frac{\Tr[X]}{d}I_d.
\end{equation}
\end{definition}
\begin{remark}
     The \emph{Haar probability distribution} is an example of a centered $1$-design.
\end{remark}
As we have described in the introduction of this subsection, our physical noise model will emerge naturally by taking the average of the effective POVM $\mathcal{A}^{\beta}$. For that, we shall assume that $Z\in\mathcal{U}(d)$ are sampled from a centered 1-design distribution.
\begin{theorem}\label{Th: general noisy two level probe}
Assume $U\in\mathbb U(A_0,A_1)$ and assume that the corresponding $Z\in\mathcal{U}(d)$ is sampled from a centered 1-design distribution $\mu$. The expectation value of the induced noisy effective POVM  $\mathcal{A}^{\beta}$  will induce a general noise model described by the two outcomes POVM $\mathbb{E}_{Z\sim\mu}[\mathcal{A}^{\beta}]$ given by:
\begin{equation*}
    \mathbb{E}_{Z\sim\mu}[\mathcal{A}^{\beta}]=(\mathbb{E}_{Z\sim\mu}[A_0^{\beta}],\mathbb{E}_{Z\sim\mu}[A_1^{\beta}]),
\end{equation*}
where $\mathbb{E}_{Z\sim\mu}[A_0^{\beta}]$ and $\mathbb{E}_{Z\sim\mu}[A_1^{\beta}]$ are given by: 
\begin{equation*}
\mathbb{E}_{Z\sim\mu}[A_0^{\beta}]=\beta_{00}\,A_0+\beta_{11}\,\frac{\Tr[A_1]}{d}I_d\quad\text{and}\quad\mathbb{E}_{Z\sim\mu}[A_1^{\beta}]=\beta_{00}\,A_1+\beta_{11}\,\frac{\Tr[A_0]}{d}I_d.
\end{equation*}
In other words, we have:  
\begin{equation*}
(\mathbb{E}_{Z\sim\mu}[A_0^{\beta}]\,,\,\mathbb{E}_{Z\sim\mu}[A_1^{\beta}]) = \beta_{00}(A_0\,,\,A_1) + \beta_{11} \left(\frac{\Tr[A_1]}{d}I_d\,,\,\frac{\Tr[A_0]}{d}I_d\right).
\end{equation*}
 
\end{theorem}
\begin{proof}
From Proposition \ref{Prop: general two level probe} we have shown that:
\begin{align}
A_0^{\beta}&=\beta_{00}\,A_0+\beta_{01}\,Z\,\sqrt{A_1}\,\sqrt{A_0}+\beta_{10}\,\sqrt{A_1}\,\sqrt{A_0}\,Z^*+\beta_{11}\,Z^*\,A_1\,Z.\\
    A_1^{\beta}&=\beta_{00}\,A_1-\beta_{01}\,Z\,\sqrt{A_1}\,\sqrt{A_0}-\beta_{10}\,\sqrt{A_1}\,\sqrt{A_0}\,Z^*+\beta_{11}\,Z^*\,A_0\,Z.
\end{align}
By taking the average on $Z$ and by linearity we have: 
\begin{align*}
\mathbb{E}_{Z\sim\mu}[A_0^{\beta}]&=\beta_{00}\,A_0+\beta_{01}\,\mathbb{E}_{Z\sim\mu}[\,Z\,\sqrt{A_1}\,\sqrt{A_0}]+\beta_{10}\,\mathbb{E}_{Z\sim\mu}[\sqrt{A_1}\,\sqrt{A_0}\,Z^*]+\beta_{11}\,\mathbb{E}_{Z\sim\mu}[Z^*\,A_1\,Z].\\
&=\beta_{00}\,A_0+\beta_{11}\,\frac{\Tr[A_1]}{d}I_d.
\end{align*}
With the same computation as the previous one, we obtain  $\mathbb E_{Z\sim\mu}[A_1^{\beta}]$ as announced in the statement of the Theorem.
\end{proof}
Having established the general description of noise models from a two-level probe model $\beta$, we shall specify two different probes: the \emph{probabilistic probe} $\sigma_t$ and the \emph{cat probe} $\gamma_t$. We will show that starting from these two different types of probes when one takes the average on $Z$ they give the \emph{same type of noisy effective POVM}, hence the same type of noise model. 
\begin{definition}\label{def: the cat probe}
We shall define the probabilistic and the cat probes are given respectively by $\sigma_t$ and $\gamma_t$ for $t\in[0,1]$ as: 
\begin{equation}
\sigma_t:=(1-t)\ketbra{0}{0}+t\ketbra{1}{1}\quad\text{and}\quad\gamma_t:=\ketbra{\lambda_t}{\lambda_t}.
\end{equation}
where $\ket{\lambda_t}$ is defined as 
\begin{equation}
\ket{\lambda_t}:=\sqrt{1-t}\ket{0}+\sqrt{t}\ket{1}.
\end{equation}
\end{definition}
As a direct consequence of Theorem \ref{Th: general noisy two level probe}, we will see in the following Corollary that the resulting noise model is independent of using a probabilistic or a cat probe.
\begin{corollary}\label{Corollary: proba or cat probe}
Assuming the evolution $U\in\mathcal{U}(2d)$ of $\rho$ coupled either to the probabilistic or the cat probe given respectively by $\sigma_t$ and $\gamma_t$. The expectation value of $\mathcal{A}^{\sigma_t}$ and $\mathcal{A}^{\gamma_t}$ are equal, and we have $\mathbb{E}_{Z\sim\mu}[\mathcal{A}^{\sigma_t}]=\mathbb{E}_{Z\sim\mu}[\mathcal{A}^{\gamma_t}]$.
\end{corollary}
\begin{proof}
The proof of the Corollary is basically that the probabilistic and the cat probe have the same diagonal elements, moreover, the off-diagonal elements will be all canceled when we take the expectation value.

By Theorem \ref{Th: general noisy two level probe} we obtain: $$\mathbb{E}_{Z\sim\mu}[A_0^{\sigma_t}]=t\,A_0+ (1-t)\,\frac{\Tr[A_1]}{d}I_d=\mathbb{E}_{Z\sim\mu}[A_0^{\gamma_t}]\,,\,\mathbb{E}_{Z\sim\mu}[A_1^{\sigma_t}]=t\,A_1+ (1-t)\,\frac{\Tr[A_0]}{d}I_d=\mathbb{E}_{Z\sim\mu}[A_1^{\gamma_t}].$$
\end{proof}
\section{The case of arbitrary POVMs}\label{sec: noise model in higher dim}
In this section, we will generalize our \emph{ physical noise model} description from an \emph{indirect measurement process} for the case of POVMs \emph{having more than two outcomes}. Here we shall consider a probe of dimension $N+1$. As in the previous section, the POVMs can be induced by indirect measurement, where the resulting POVMs are of $N+1$ outcomes. The noise in the probe preparation will induce noisy POVMs of $N+1$ outcomes. 

The unitary interaction between the system and the probe will induce and give the form of the POVMs. This will allow us, to introduce the set of the unitaries that give rise to the POVMs. With the noisy probe and set of all the unitaries that we shall consider as random, a natural noise model emerges by taking the average over the unitaries.  

In Subsection \ref{subsec: General induced POVM}, we will give the induced POVMs in the perfect situation. In Subsection \ref{subsec: General induced POVM}, we will give our physical noise model description. 

\subsection{General induced POVM}\label{subsec: General induced POVM}

In the following subsection, we will give a general model where we will describe an emergent POVM $\mathcal{A}$ with $N+1$ outcomes. For that, as in Section \ref{sec: two level model}, we shall assume the probe is initially prepared in $\ketbra{0}{0}$ and the total system is given by $\rho\otimes\ketbra{0}{0}$. The dynamics are given by unitary matrices $U\in\mathcal U ((N+1)d)$. Explicitly the evolution of the total system, as usual, is given by:
\begin{equation*}
\rho\otimes\ketbra{0}{0}\to U\Big(\rho\otimes\ketbra{0}{0}\Big)\,U^*,
\end{equation*}
with $U\in\mathcal{U}((N+1)d)\subseteq \mathcal{M}_d\otimes\mathcal{M}_{N+1}.$

We denote by $\mathbb P(i)$ the probability of obtaining the outcome $i\in[0,N]$ by measuring the probe on the evolved total system.
In the following proposition, we will show that the measurement over the probe will induce an effective emergent POVM $\mathcal{A}$ with $ N+1$ outcomes.
\begin{proposition}\label{prop: effective general POVM}
The probability of obtaining the outcome $i\in[0,N]$ induce  an effective POVM $\mathcal{A}$ with $N+1$ outcomes given by: 
\begin{equation}
    \mathcal A:=(A_0,\cdots,A_N)\quad \text{with} \quad  A_i:=U^*_{i0}\,U_{i0},\quad \forall i\in[0,N].
\end{equation}
\end{proposition}
\begin{proof}
An explicit computation of $\mathbb P(i)$ shows that: 
\begin{align*}
\mathbb{P}(i)&=\Tr\Big[U\,\Big(\rho\otimes\ketbra{0}{0}\Big)\,U^*\,\Big(I_d\otimes\ketbra{i}{i}\Big)\Big]\\
	          &=\Tr\Big[U^*_{i0}\,U_{i0}\,\rho\Big]=\Tr\Big[A_i\rho\Big].
        \end{align*}
        Where we have defined $A_i:= U^*_{i0}\,U_{i0}\geq0$.
        
        Therefore $A_i\geq0$ for all $i\in[0,N]$ and by using the unitarity of $U$ one can check easily that $\sum_{i=0}^N\,A_i=I_d$ holds, hence $\mathcal{A}:=(A_0,\cdots,A_N)$ defines a POVM.
\end{proof}
The proposition above motivates the  following definition, where we will define a set of unitaries $U\in\mathcal{U}((N+1)d)$ where for fixed $U_{i0}$ for all $i\in[0,N]$ is given by $U^*_{i0}\,U_{i0}=:A_i$. 
\begin{definition}\label{def: General unitaries}
Consider a POVM $\mathcal{A}=(A_0,\cdots,A_N)$. We define a subset of unitary matrices $\mathbb{U}\Big(\{A_i\}_{i\in[0,N]}\Big)$ as:
\begin{equation}\label{def: general unitary set}
\mathbb{U}\Big(\{A_i\}_{i\in[0,N]}\Big):=\Big\{U\in\mathcal{U}((N+1)d)\Big|\,U^*_{i0}\,U_{i0}=A_i\Big\}.
\end{equation}
\end{definition}

\begin{remark}
    It is a very difficult task to explicitly characterize the set of unitaries given above, in comparison with the set given in Definition \ref{def: the unitary set} which was fully characterized in Proposition \ref{prop: unitray set}. 
\end{remark}

\begin{remark}
The set $\mathbb U\Big(\{A_i\}_{i\in[0,N]}\Big)$ is non-empty.
To show this, we shall show it is possible to complete the matrix $U$ having as it first block column $U_{i0}=\sqrt{A_i}$. 
This is possible as long as the first $d$ columns of $U$ are orthonormal. For that, let two different columns $c_1$, $c_2$ of $U$, where $c_j\in[d]$ and $j\in\{1,2\}$. The scalar product of two columns is given by: 
\begin{equation*}
\sum_{l=1}^d\sum_{i=0}^N\Big(\sqrt{A_i}\,(l,c_1)\Big)^*\,\sqrt{A_i}(l,c_2)=\sum_{l=1}^d\sum_{i=0}^N A_i(c_1,c_2)=\delta_{c_1,c_2},
\end{equation*}
where in the equality above we have used that $\sum_{i=0}^N A_i=I_d$. 

Therefore, the first $d$ columns of $U$ are orthonormal, hence we can always complete a unitary matrix $U\in\mathbb U\Big(\{A_i\}_{i\in[0,N]}\Big)$. 

\end{remark}
\begin{remark}
The set $\mathbb{U}\Big(\{A_i\}_{i\in[0,N]}\Big)$ is invariant invariant under the action of element in $\mathcal{U}^{\times2N+1}(d)$ defined as below.

Let the action $\circ':\mathcal{U}(d)^{\times2N+1}\curvearrowright\mathbb{U}\Big(\{A_i\}_{i\in[0,N]}\Big)$ defined by: 
\begin{align*}
    \circ':\mathcal{U}(d)^{\times2N+1}\times\mathbb{U}\Big(\{A_i\}_{i\in[0,N]}\Big)&\to \mathbb{U}\Big(\{A_i\}_{i\in[0,N]}\Big),\\
    \Big((V_0,\cdots,V_N,Z_1,\cdots,Z_N),U\Big)&\to\mathbb{U}\Big(\{A_i\}_{i\in[0,N]}\Big).
\end{align*}
For $U\in \mathbb{U}\Big(\{A_i\}_{i\in[0,N]}\Big)$ the action of $\circ'$ is defined by: 
\begin{equation*}
    U'=(V_0,\cdots,V_N,Z_1,\cdots,Z_N)\circ' U:=\begin{bmatrix}
    V_0 & & \\
    & V_1\\
    & & \ddots & \\
    & & & V_N
  \end{bmatrix}
  \,U\,
  \begin{bmatrix}
    I_d& & \\
    & Z_1^*\\
    & & \ddots & \\
    & & & Z_N^*
  \end{bmatrix}.
\end{equation*}
Remark from a simple computation that:
\begin{equation*}
U_{i0}'^*U_{i0}'=U_{i0}^*V_i^*V_iU_{i0}=A_i,
\end{equation*}
therefore we have $U'\in \mathbb{U}\Big(\{A_i\}_{i\in[0,N]}\Big)$.
\end{remark}

The Definition \ref{def: General unitaries} will play a crucial role in the rest of this section to introduce our physical noise model.

\subsection{Physical noise model}\label{subsec: General noise model}
In the following subsection, we will introduce our main general \emph{physical noise} model. For that, we will assume initially that the quantum state $\rho$ is coupled to a probe given by an $(N+1)$-level quantum state $\beta\in\mathcal{M}_{N+1}^{1,+}$. 
\begin{definition}
    Let $\beta\in\mathcal{M}_{N+1}^{1,+}$ the $N+1$ level quantum state describing the probe: 
    \begin{equation*}
        \beta=\sum_{i,j=0}^N\beta_{i,j}\ketbra{i}{j} \in\mathcal{M}_{N+1}^{1,+}.
    \end{equation*}
\end{definition}
Assuming the evolution of the total system, quantum state $\rho$ coupled to a probe $\beta$ is given by: 
\begin{equation*}
    \rho\otimes\beta\to U\Big(\rho\otimes\beta\Big)U^*,
\end{equation*}
with $U\in \mathbb U\Big(\{A_i\}_{i\in[0,N]}\Big)$.
We denote by $\mathbb P_{\beta}(i)$ the probability of obtaining an outcome $i\in[0,N]$ by measuring the probe on the evolved total system.

In the following we will show the outcome probability of obtaining $i\in[0,N]$ will induce an effective noisy POVM $\mathcal{A}^{\beta}$.
 
\begin{proposition}\label{prop: effective noisy general POVM}
The probability of obtaining the outcome $i\in[0,N]$ by measuring the probe on the evolved total system, will induce an effective POVM $\mathcal{A}^{\beta}$ where: 
\begin{equation*}
    \mathcal{A}^{\beta}=(A_0^{\beta},\cdots,A_N^{\beta}),\quad\text{and}\quad A_i^{\beta}=\sum_{c,k=0}^N\,\beta_{ck}\,U^*_{ik}\,U_{ic},\quad \forall i\in[0,N].
\end{equation*}
\end{proposition}

\begin{proof}
An explicit computation of $\mathbb{P}_{\beta}(i)$ shows that:
\begin{align*}
    \mathbb P_{\beta}(i)&=\Tr\Big[U\,(\rho\otimes\beta)\,U^*\,(I_d\otimes\ketbra{i}{i})\Big]\\
    &=\sum_{c,k=0}^N\beta_{ck}\,\Tr\Big[U^*_{ik}\,U_{ic}\,\rho\Big]\\
    &=\Tr[A_i\rho], 
\end{align*}
where we have defined $A_i^{\beta}:=\sum_{c,k=0}^N\beta_{ck}\,U^*_{ik}\,U_{ic}\geq 0$ and $\sum_{i=0}^N A_i^{\beta}=\Tr\beta=1$, therefore $\mathcal{A}^{\beta}:=(A_0^{\beta},\cdots,A_N^{\beta})$ defines a POVM.
\end{proof}
To introduce our physical noise model in this general setting, we shall assume that the unitaries $U\in\mathbb{U}\Big(\{A_i\}_{i\in[0,N]}\Big)$ are randomly sampled from a \emph{nice probability} measure.

\begin{definition}\label{def: nice proba measure}
A probability distribution $\mu$ on $\mathbb{U}\Big(\{A_i\}_{i\in[0,N]}\Big)$ is a \emph{nice probability measure} if it is invariant under right multiplication with a unitary $Z\in\mathcal{U}((N+1)d)$ of the form $Z:=\sum_{i=0}^N\,Z_i\otimes\ketbra{i}{i}\quad\text{with}\quad Z_0=I_d\quad\text{and}\quad Z_i\in\mathcal{U}(d)$ for all $i\neq0.$
 More precisely, we have for a random unitary $U\in\mathbb{U}\Big(\{A_i\}_{i\in[0,N]}\Big)$:
$$\mathbb{E}_{U\sim\mu}[f(U)]=\mathbb{E}_{U\sim\mu}[f(U\cdot Z)].$$
\end{definition}

\begin{remark}
    To construct an example of a nice probability measure, consider a fixed element $U_0 \in \mathbb{U}(\{A_i\}_{i\in[0,N]})$. Let also $V_1, \ldots, V_N \in \mathcal U(d)$ be independent Haar-distributed random unitary matrices. Define the random variable 
    $$U:= U_0 (I_d \oplus V_1 \oplus \cdots \oplus V_N).$$
    We claim that $U$ has a nice distribution. Indeed, if $Z_1, \ldots, Z_N \in \mathcal U(d)$ are arbitrary unitary matrices,
    $$\mathbb E[f(U\cdot Z) = \mathbb E [f(U_0 V Z)] = \mathbb E [f(U_0 \tilde V)] = \mathbb E[f(U)],$$
    where we have used the fact that $V$ and $\tilde V = V Z$ have the same (block-Haar) distribution.
\end{remark} 

\begin{proposition}\label{prop: nice proba null}
    Let $\mu$ be a nice probability measure on $\mathbb{U}\Big(\{A_i\}_{i\in[0,N]}\Big)$. We have: 
    \begin{equation}
        U\in \mathbb{U}\Big(\{A_i\}_{i\in[0,N]}\Big),\quad\mathbb E_{U\sim\mu}[U^*_{ik}\,U_{i0}]=0.
    \end{equation}
\end{proposition}
\begin{proof}
    By using the right invariance property of the nice measure $\mu$ from Definition \ref{def: nice proba measure} we have: 
    \begin{equation*}
        \mathbb E_{U\sim\mu}[U^*_{ik}\,U_{i0}]=Z^*_k\,\mathbb E_{U\sim\mu}[U^*_{ik}\,U_{i0}].
    \end{equation*}
    Since this last equation should hold for all $Z_k\in U(d)$, the proposition holds.
\end{proof}
The physical motivation of introducing the definition above, all the microscopical degrees of freedom generating the dynamics of the total system are unknown. Therefore, defining them as random becomes relevant. However, assuming they are sampled from a nice probability measure will become clear in the following theorem. In the following, we will show that when one averages over the unitaries an effective noise model emerges.

\begin{theorem}\label{Th: Expectation value of general noisy POVM}
Let $U$ be a random unitary operator sampled from a nice probability measure $\mu$, and assume that the elements of the probe density matrix $\beta_{cc}$ are constant for all $c\neq0$. The averaged POVM given by $\mathbb E_{U\sim\mu}[\mathcal{A}^{\beta}]$ defines a noise model given by: 
 \begin{equation*}
\mathbb E_{U\sim\mu}[A^{\beta}_i]=\beta_{00}\,A_i+\frac{1-\beta_{00}}{N-1}\Big(1-\frac{1}{d}\Tr[A_i]\Big)\,I
\end{equation*}
One can interpret the element $\beta_{00}$ of the probe $\beta\in\mathcal{M}^{1,+}_{N+1}$ as a noise parameter. 
\end{theorem}
\begin{proof}
Let $\mu$ a nice probability measure on $\mathbb{U}\Big(\{A_i\}_{i\in[0,N]}\Big)$, we recall from Proposition \ref{prop: effective noisy general POVM} that:  $$A_i^{\beta}=\sum_{c,k=0}^N\beta_{ck}\,U^*_{ik}\,U_{ic}.$$
An explicit computation of the expectation value of $\mathcal{A}^{\beta}$ on $U\sim\mu$ yields to:
\begin{align*}
    \mathbb{E}_{U\sim\mu}[A_i^{\beta}]&=\sum_{c,k=0}^N\beta_{ck}\,\mathbb E_{U\sim\mu}[U^*_{id}\,U_{ic}]\\
    &=\beta_{00}\,A_i+\sum_{k=0}^N\beta_{0k}\,\mathbb E_{U\sim\mu}[U^*_{ik}\,U_{i0}]+\sum_{c=0}^N\beta_{c0}\,\mathbb E_{U\sim\mu}[U^*_{i0}\,U_{ic}]+\sum_{c,k\neq0}^N\beta_{ck}\,\mathbb E_{U\sim\mu}[U^*_{ik}\,U_{ic}]\\
    &=\beta_{00}\,A_i+\sum_{k=0}^N\beta_{0k}\,Z_k^*\,\mathbb E_{U\sim\mu}[U^*_{ik}\,U_{i0}]+\sum_{c=0}^N\beta_{c0}\,\mathbb E_{U\sim\mu}[U^*_{i0}\,U_{ic}]\,Z_c+\sum_{c,k\neq0}^N\beta_{ck}\,Z^*_k\mathbb E_{U\sim\mu}[U^*_{ik}\,U_{ic}]\,Z_c\\
    &=\beta_{00}\,A_i+\sum_{k=0}^N\beta_{0k}\,Z_k^*\,\underbrace{\mathbb E_{U\sim\mu}[U^*_{ik}\,U_{i0}]}_{0}+\sum_{c=0}^N\beta_{c0}\,\underbrace{\mathbb E_{U\sim\mu}[U^*_{i0}\,U_{ic}]}_{0}\,Z_c+\sum_{c,k\neq0}^N\beta_{ck}\,Z^*_k\underbrace{\mathbb E_{U\sim\mu}[U^*_{ik}\,U_{ic}]}_{\delta_{ck}\,\Gamma_{ic}\,I}\,Z_c.
    \end{align*}
In the last line, the simplification is obtained by the use of Proposition \ref{prop: nice proba null}, therefore we have the final expression given by: 
    \begin{equation}
\mathbb{E}_{U\sim\mu}[A_i^{\beta}]=\beta_{00}\, A_i+\sum_{c\neq0}^N\beta_{cc}\,\Gamma_{ic}\,I.
\end{equation}
where we have defined the constant $\Gamma_{ic}$ as:
\begin{equation}\label{eq: Gamma_ic}
\Gamma_{ic}:=\frac{1}{d}\Tr\Big[\mathbb E_{U\sim\mu}[U^*_{ic}\,U_{ic}]\Big].
\end{equation}
One can check with the definition of $\Gamma_{ic}$ that $\mathbb{E}_{U\sim\mu}[A_i^{\beta}]=I$ is a POVM, from the expression of $\mathbb{E}_{U\sim\mu}[A_i^{\beta}]$ the positivity condition holds, the only remaining condition is to show that the $\sum_i \mathbb{E}_{U\sim\mu}[A_i^{\beta}]$.

By using that $U^*U=I_{(N+1)d}$ we have:
\begin{equation}\label{eq: sum_i gammaic}
U^*U=I_{(N+1)d}\implies\sum_{i=0}^NU^*_{ic}\,U_{ic}=I_d\implies\sum_{i=0}^N\Gamma_{ic}=1.\end{equation}
Then we have from equation \eqref{eq: sum_i gammaic} that: 
\begin{equation}
   \sum_{i=0}^N \mathbb{E}_{U\sim\mu}[A_i^{\beta}]=\beta_{00}\, \sum_{i=0}^N\,A_i+\sum_{c\neq0}^N\beta_{cc}\sum_{i=0}^N\,\Gamma_{ic}\,I=\Tr[\beta]=1.
\end{equation}

By assuming for all $\forall c\neq0$, $\beta_{cc}$ are constants and do not depend on $c$, we have that:
\begin{align*}
    \mathbb{E}_{U\sim\mu}[A_i^{\beta}]&=\beta_{00}\, A_i+\sum_{c\neq0}^N\beta_{cc}\,\Gamma_{ic}\,I\\
    &=\beta_{00}\,A_i+\frac{(1-\beta_{00})}{N-1}\sum_{c\neq0}^N\,\Gamma_{ic}\,I\\
    &=\beta_{00}\,A_i+\frac{(1-\beta_{00})}{N-1}\Big(1-\frac{1}{d}\Tr[A_i]\Big)I.
\end{align*}
To obtain the last equality we have used the unitary property of $U$ where we have:
\begin{equation}
UU^*=I_{(N+1)d}\implies\sum_{c=0}^NU_{ic}U^*_{ic}=I_d\implies\sum_{c=0}^N\Gamma_{ic}=1,
\end{equation}
and by combining the expression from equation \eqref{eq: Gamma_ic},we deduce that:
\begin{equation}
\sum_{c\neq0}^N\Gamma_{ic}=1-\Gamma_{i0}=1-\frac{1}{d}\Tr\Big[\mathbb E_{U\sim\mu}[U^*_{i0}\,U_{i0}]\Big]=1-\frac{1}{d}\Tr[A_i].
\end{equation}
\end{proof}

\begin{remark}
    An interesting open question is to explore the possible noise models that can be obtained by considering different probe states $\beta$ (with non-constant diagonal) and / or different probability distributions for interaction unitary $U$. Could the noise models from \cite{designolle2019incompatibility} be obtained in this way? 
\end{remark}

\section{Applications to compatibility}\label{sec: examples}

In this section, we shall give an application of the physical noise model we introduced in the setting of compatibility of quantum measurements. As we have described in Section \ref{sec:noisy-POVMs}, one can always make incompatible measurements compatible by adding noise. A natural notion that we shall encounter in this section is the notion of \emph{robustness}. Robustness captures the minimal amount of noise needed to make incompatible measurements compatible. We shall give the formulation of robustness in the case of three different noise models, the uniform noise model, the depolarizing noise model, and our physical noise model. We shall give the SDP and its dual formulation in the three types of noise models mentioned above. We will end this section by giving an explicit example of robustness in the case of a particular pair of measurements, using the three noise models.

In general, deciding whether a set of quantum measurements is (in-)compatible is a hard task, which can be formulated as a semidefinite program (SDP). These are a type of convex optimization programs with positive semidefinite constraints (see \cite{boyd2004convex} for a general introduction) that will allow us later to give a formulation of robustness.

 We recall that a noise model is given by a convex combination of a given POVM $\mathcal{A}$ with a trivial POVM $\mathcal{T}$. We recall from Section \ref{sec:noisy-POVMs} the uniform noise and the depolarizing noise are given by:
 \begin{itemize}
     \item The uniform noise model for $\alpha\in[0,1]$ is defined by: 
     \begin{equation*}         \mathcal{A}\to\mathcal{A}^{\alpha}:=\alpha \mathcal{A}+(1-\alpha)\mathcal{T}, 
     \end{equation*}
     where $\mathcal{T}:=(t_1 I,\cdots,t_N I)$, with $t_i:=1/N$ for all $i\in[N]$.
     \vspace{5mm}
     \item The depolarizing noise for $\alpha\in[0,1]$ is defined by:
     \begin{equation*}         \mathcal{A}\to\mathcal{A}^{\alpha}:=\alpha \mathcal{A}+(1-\alpha)\mathcal{T}, 
     \end{equation*}
     where $\mathcal{T}:=(t_1 I,\cdots,t_N I)$, with $t_i:=\Tr A_i/d$ for all $i\in[N]$.
     \vspace{5mm}
     \end{itemize}
We shall also recall that our main noise model is based on an indirect measurement process obtained in Theorem \ref{Th: Expectation value of general noisy POVM} where:
     \begin{itemize}
         \item For a given probe $\beta$ where $\beta_{00}$ can be interpreted as a noise parameter taking values in $[0,1]$, \begin{equation*}
         \mathcal{A}\to\mathcal{A}^{\beta}=\beta_{00}\mathcal{A}+(1-\beta_{00})\mathcal{T},
     \end{equation*}
     and the trivial POVM is given by $\mathcal{T}=(t_0I,\cdots,t_N I)$, with $t_i=\frac{1}{N-1}\Big(1-\frac{1}{d}\Tr[A_i]\Big)$ for all $i\in[0,N]$. 
 \end{itemize}
 Note that in the formulas above, it is the quantity $1-\alpha$, resp.~$1-\beta_{00}$, which measures the amount of noise present in the POVM $\mathcal A^\alpha$. In the following, we shall give the definition of \emph{incompatibility robustness} that computes the minimal amount of noise needed to make two POVMs compatible, see \cite{heinosaari2016invitation,designolle2019incompatibility} and the references therein. 
\begin{definition}\label{def: robustness}
    For two POVMs $\mathcal{A}$ and $\mathcal{B}$, and a noise model $\mathcal{T}$, the robustness of incompatibility of $\mathcal{A}$ and $\mathcal{B}$ is defined by: 
    \begin{equation}\label{eq:incompatibility-robustness}
     \alpha^*:=\sup_{\alpha\in[0,1]}\{\alpha\, :\, \mathcal{A}^{\alpha} \text{ and }\mathcal{B}^{\alpha}\text{ are compatible}\}.
     \end{equation} 
\end{definition}
For a noise model given by a trivial operator $\mathcal{T}$, we can give the SDP formulation of robustness. Therefore, we will recall the SDP and its dual formulation for the uniform and depolarizing noise. 
\begin{definition}\label{def: rob uniform depol}
    The SDP of the robustness of two POVMs with $N$ outcomes $\mathcal{A}$ and $\mathcal{B}$ with a uniform and depolarizing noise are given respectively by: 
    \begin{equation*}\label{eq: primal op uniform and dep}
\alpha^*_{u}:=\begin{cases}\max_{\alpha\in[0,1],\{C_{ij}\}}\alpha\quad\text{s.t}&\\
        \sum_{j=1}^NC_{ij}=\alpha A_i+(1-\alpha)I/N&\\
        \sum_{i=1}^NC_{ij}=\alpha B_j+(1-\alpha)I/N&  
        \end{cases}
        \quad\text{and}\quad 
\alpha^*_{d}:=\begin{cases}\max_{\alpha\in[0,1],\{C_{ij}\}}\alpha\quad\text{s.t} &\\
        \sum_{j=1}^NC_{ij}=\alpha A_i+(1-\alpha)\Tr A_i/d&\\
        \sum_{i=1}^NC_{ij}=\alpha B_j+(1-\alpha)\Tr B_j/d&  
        \end{cases},
    \end{equation*}
 where the $\alpha^*_u$ and $\alpha^*_d$ stand respectively for the case of the uniform and the depolarizing noise.
\end{definition}
\begin{proposition}\cite{designolle2019incompatibility}\label{prop: dual rob uniform depol}
    The dual SDP of the robustness of the case of the uniform and the depolarising noise for two POVMs with $N$ outcomes are given respectively by: 
    \begin{align*}
        \alpha^*_u&=\begin{cases}
            \min_{\{X_i\}_{i\in[N]},\{Y_j\}_{j\in[N]}} 1+\sum_{i=1}^N \Tr[X_iA_i]+\sum_{j=1}^N\Tr[Y_jB_j]\quad\text{s.t}&\\
            X_i+Y_j\geq 0\,\text{for all}\quad i,j\in[N]\quad\text{and}&\\
            1+\sum_{i=1}^N \Tr[X_iA_i]+\sum_{j=1}^N\Tr[Y_jB_j]\geq \sum_{i=1}^N\Tr X_i/N + \sum_{j=1}^N\Tr Y_j/N
        \end{cases}\\
        \alpha^*_d&=\begin{cases}
            \min_{\{X_i\}_{i\in[N]},\{Y_j\}_{j\in[N]}} 1+\sum_{i=1}^N \Tr[X_iA_i]+\sum_{j=1}^N\Tr[Y_jB_j]\quad\text{s.t}&\\
            X_i+Y_j\geq 0\,\text{for all}\quad i,j\in[N]\quad\text{and}&\\
            1+\sum_{i=1}^N \Tr[X_iA_i]+\sum_{j=1}^N\Tr[Y_jB_j]\geq \sum_{i=1}^N\Tr X_i\,\Tr X_i/d + \sum_{j=1}^N\Tr Y_j\,\Tr B_j/d
        \end{cases}
    \end{align*}
\end{proposition}
\begin{proof}
Let us consider the following Lagrangian $\mathcal{L}_u$ corresponding to the optimization problem of robustness for the uniform noise given by the primal SDP in Eq.~\eqref{eq: primal op uniform and dep}: 
\begin{equation*}
\mathcal{L}_u:=\alpha+\sum_{i,j=1}^N\Big\langle C_{ij},X_{ij}\Big\rangle-\sum_{i=1}^N\Big\langle\sum_{j=1}^NC_{ij}-\alpha A_i-(1-\alpha)\frac{I}{N},X_i\Big\rangle-\sum_{j=1}^N\Big\langle\sum_{i=1}^NC_{ij}-\alpha B_j-(1-\alpha)\frac{I}{N},Y_j\Big\rangle
\end{equation*}
where $X_{ij}$, $X_i$ and $Y_j$ are positive semidefinite matrices for all $i,j\in[N]$ encoding the different constrains of the primal optimisation problem. 

Since Slater's condition is satisfied, the minimax duality holds, therefore we have: 
\begin{equation}
    \max_{\alpha,\{C_{ij}\}}\quad\inf_{\{X_{ij}\},\{X_i\},\{Y_j\}}\mathcal{L}_u=\inf_{\{X_{ij}\},\{X_i\},\{Y_j\}}\quad\max_{\alpha,\{C_{ij}\}}\mathcal{L}_u,
\end{equation}
By expanding the expression of $\mathcal{L}_u$ and taking the maximum over $\{C_{ij}\}$ and $\alpha$ we obtain: 
\begin{equation*}
    \max_{\{C_{ij}\},\alpha}\mathcal{L}_u=
        \sum_{i=1}^N\Tr X_i/N+\sum_{j=1}^N\Tr Y_j/N
\end{equation*}
if the following constraints hold
\begin{align*}
     1+\sum_{j=1}^N\Big\langle B_j,Y_j\Big\rangle+\sum_{i=1}^N\Big\langle A_i,X_i\Big\rangle&\geq\sum_{i=1}^N\Tr X_i/N+\sum_{j=1}^N\Tr Y_j/N\quad\text{and}\\
        X_{ij}&=X_i+Y_j\geq0,
\end{align*}
and $+\infty$ if they do not. Above we have used the Hilbert-Schmidt scalar $\langle A,B\rangle_{\text{H.S}}=\Tr (A^*B)$. 
By plugging the obtained value in the expression of $\max_{\{C_{ij}\},\alpha}\mathcal{L}_u$ we obtain the result of the statement. 

The same computation can be explicitly done for the depolarising noise, where it is easy to check by duality the results of the statement in the proposition with the following Lagrangian:   
\begin{align*}
\mathcal{L}_d:=\alpha+\sum_{i,j=1}^N\Big\langle C_{ij},X_{ij}\Big\rangle&-\sum_{i=1}^N\Big\langle\sum_{j=1}^NC_{ij}-\alpha A_i-(1-\alpha)\frac{\Tr A_i}{d},X_i\Big\rangle\\
&-\sum_{j=1}^N\Big\langle\sum_{i=1}^NC_{ij}-\alpha B_j-(1-\alpha)\frac{\Tr B_j}{N},Y_j\Big\rangle.
\end{align*}
\end{proof}
In the following, we shall give the SDP formulation of robustness in the case of our physical noise model and its dual formulation.
\begin{definition}\label{def: SDP rob physical model}
    The SDP of the robustness of two POVMs with $N+1$ outcomes $\mathcal{A}$ and $\mathcal{B}$ in the physical noise model is given by: 
    \begin{equation*}
        \alpha^*_p:=\begin{cases}
            \max_{\beta_{00}\in[0,1],\{C_{ij}\}} \beta_{00}\quad\text{s.t}&\\
            \sum_{j=0}^N C_{ij}=\beta_{00} A_i+\frac{(1-\beta_{00})}{N-1}\Big(1-\frac{\Tr A_i}{d}\Big)I&\\
            \sum_{i=0}^N C_{ij}=\beta_{00} B_j+\frac{(1-\beta_{00})}{N-1}\Big(1-\frac{\Tr B_j}{d}\Big)I&
        \end{cases},
    \end{equation*}
    where $\alpha^*_p$ stands for robustness in our physical model. 
\end{definition}

We give next the dual formulation. 

\begin{proposition}\label{prop: dual SDP rob physical model}
    The dual SDP of robustness in the case of our physical noise model is given by: 
    \begin{equation*}
    \alpha^*_p=
            \min_{\{X_i\}_{i\in[0,N]},\{Y_j\}_{j\in[0,N]}} 1+\sum_{i=0}^N\Tr[X_iA_i]+\sum_{j=0}^N\Tr[Y_jB_j]
            \end{equation*}
            with the following constraints: 
            
            \begin{align*}
            X_i+Y_j&\geq 0\,\text{for all}\quad i,j\in[0,N]\,\text{and}\\
            1+\sum_{i=0}^N \langle X_iA_i\rangle+\sum_{j=0}^N\Tr[Y_jB_j]&\geq \frac{1}{N-1}\Big(\sum_{i=0}^N\Tr X_i +            \sum_{j=0}^N\Tr Y_j-\frac{1}{d}\sum_{j=0}^N\Tr B_j\,\Tr Y_j-\frac{1}{d}\sum_{i=0}^N\Tr X_i\,\Tr A_i\Big).
        \end{align*}        
\end{proposition}
\begin{proof}
Let us define the Lagrangian $\mathcal{L}_p$ corresponding to the optimization problem of robustness in our main physical model given by:
    \begin{align*}
    \mathcal{L}_p:=\beta_{00}+&\sum_{i,j=0}^N\Big\langle C_{ij},X_{ij}\Big\rangle-\sum_{i=0}^N\Big\langle \sum_{j=0}^N C_{ij}-\beta_{00} A_i-\frac{(1-\beta_{00})}{N-1}\Big(1-\frac{\Tr A_i}{d}\Big)I,X_i\Big\rangle\\
    &-\sum_{j=0}^N\Big\langle \sum_{i=0}^N C_{ij}-\beta_{00} B_j-\frac{(1-\beta_{00})}{N-1}\Big(1-\frac{\Tr B_j}{d}\Big)I,Y_j\Big\rangle.
    \end{align*}
Due to Slater's condition which is satisfied, the minimax duality holds, therefore we have: 
\begin{equation*}
    \max_{\beta_{00},\{C_{ij}\}}\quad\inf_{\{X_{ij}\},\{X_i\},\{Y_j\}}\mathcal{L}_p=\inf_{\{X_{ij}\},\{X_i\},\{Y_j\}}\quad\max_{\beta_{00},\{C_{ij}\}}\mathcal{L}_p.
\end{equation*}
By expanding the expression of $\mathcal{L}_p$ and taking the maximum of $\{C_{ij}\}$ and $\beta_{00}$:
\begin{equation*}
    \max_{\{C_{ij}\},\beta_{00}}\mathcal{L}_p= \frac{1}{N-1}\Big(\sum_{i=0}^N\Tr X_i +            \sum_{j=0}^N\Tr Y_j-\frac{1}{d}\sum_{j=0}^N\Tr B_j\,\Tr Y_j-\frac{1}{d}\sum_{i=0}^N\Tr X_i\,\Tr A_i\Big),
\end{equation*}
with the following constraints: 
\begin{equation*}
1+\sum_{i=0}^N \langle X_iA_i\rangle+\sum_{j=0}^N\Tr[Y_jB_j]\geq \frac{1}{N-1}\Big(\sum_{i=0}^N\Tr X_i +            \sum_{j=0}^N\Tr Y_j-\frac{1}{d}\sum_{j=0}^N\Tr B_j\,\Tr Y_j-\frac{1}{d}\sum_{i=0}^N\Tr X_i\,\Tr A_i\Big),
\end{equation*}
and \begin{equation*}    X_{ij}=X_i+Y_j\geq 0\,\text{for all}\quad i,j\in[0,N],
\end{equation*}
and $+\infty$ if the constraints above are not satisfied. 

By plugging the obtained value in the expression of $\max_{\{C_{ij}\},\alpha}\mathcal{L}_p$ we obtain the result of the statement. 
\end{proof}

\bigskip

Let us now compare the incompatibility robustness of a pair of measurements using the three different noise models: uniform, depolarizing, and our physical model. We consider an example of two POVMs on $\mathbb C^3$ with two outcomes $\mathcal{A}=(A_0,A_1)$ and $\mathcal{B}=(B_0,B_1)$ given by: 
\begin{figure}[httb]
    \centering
\includegraphics[width=.47\textwidth]{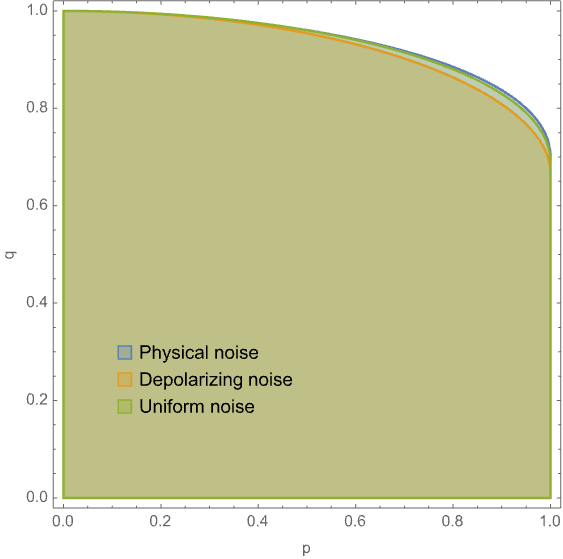}\quad
\includegraphics[width=.47\textwidth]{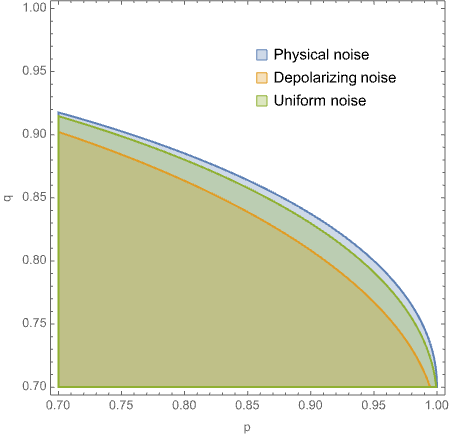}
    \caption{Compatibility region of the POVMs $\mathcal A$, $\mathcal B$ for the different noise models: the uniform, the depolarizing, and our physical noise model. On the right panel, we zoom in on the parameter region where the noise models yield different results.}
    \label{fig: compatibility region, uni dep and our noise}
\end{figure}

\begin{equation*}
    A_0:=\begin{bmatrix}
1/3&0&0\\
0&2/3&0\\
0&0&0\end{bmatrix},\,
    A_1:=\begin{bmatrix}
2/3&0&0\\
0&1/3&0\\
0&0&1\end{bmatrix},
\end{equation*}
and \begin{equation*}
    B_0=\ketbra{f_1}{f_1}+\ketbra{f_2}{f_2}\quad\text{and}\quad B_1=\ketbra{f_3}{f_3},
\end{equation*}
where $\{\ket{f_i}\}$ for $i\in[3]$ are the columns of the Fourier matrix 
	$$F_3 = \frac{1}{\sqrt 3}\begin{bmatrix}
	1 & 1 & 1 \\
	1 & \omega & \omega^2 \\
	1 & \omega^2 & \omega
	\end{bmatrix},$$
	with $\omega = \exp(2 \pi \mathrm{i}/3)$.
 We plot the \emph{compatibility region} of the POVMs defined above in Figure \ref{fig: compatibility region, uni dep and our noise}. Recall that the compatibility region \cite{bluhm2018joint} is a generalization of the incompatibility robustness from Eq.~\eqref{eq:incompatibility-robustness}:
 $$\Gamma(\mathcal A, \mathcal B) := \{(p,q) \in [0,1]^2 \, : \,  \mathcal A^p \text{ and } \mathcal B^{q} \text{ are compatible}\}.$$
 Notice in the formula above that the region $\Gamma(\mathcal A, \mathcal B)$ depends on the noise model used to define $\mathcal A^p$ and $\mathcal B^{q}$.  The compatibility regions corresponding to the different noise models are approximately the same; small differences can be observed for noise parameters in the interval $[0.7,1]$. One can notice that with our physical noise model, we require \emph{less} noise in order to make the POVMs $\mathcal A$ and $\mathcal B$ compatible with the other noise models considered in the literature.

\section{conclusion}

    In this work, we have proposed a new point of view for characterising the noise effect in measurement devices. With an indirect measurement process, by coupling the quantum state with a probe and assuming the dynamic is random, we have shown that a natural noise model emerges. The effect of noise in our point of view is a consequence of the lack of preparation of the probe. We have analysed in detail the case of the two-level probe system and a general $N+1$ level probe system, corresponding to an indirect measurement scheme with $N+1$ outcomes. Moreover, we provided an application for the compatibility of quantum measurements.  One should also mention that our work is not an attempt to explain the measurement problem in quantum theory. However, it sheds new light on the noise effects considered previously in the literature, constructing, for the first time, a noise model with a physical interpretation.

    In this work we have assumed that the dynamic is given by random unitary that will couple the probe and the state. This is motivated by the assumed unknown microscopic degrees of freedom. One can ask, if one can go beyond this simplified scenario to a more concrete one with a Hamiltonian description. The Hamiltonian of the probe coupled to the system which may be also random. One can ask what kind of noise model will emerge when one measures the probe and takes the average over the randomness which we shall investigate for future work.

\bigskip

\noindent\textit{Acknowledgements.} F.L.~would like to thank Denis Rochette for help with Mathematica code. The authors were supported by the ANR project \href{https://esquisses.math.cnrs.fr/}{ESQuisses}, grant number ANR-20-CE47-0014-01 as well as by the PHC program  \emph{Star} (Applications of random matrix theory and abstract harmonic analysis to quantum information theory). I.N.~has also received support from the ANR project \href{https://www.math.univ-toulouse.fr/~gcebron/STARS.php}{STARS}, grant number ANR-20-CE40-0008. C.P~is also supported by the ANR projects Q-COAST ANR- 19-CE48-0003, ``Quantum Trajectories'' ANR-20-CE40-0024-01, and ``Investissements d'Avenir'' ANR-11-LABX-0040 of the French National Research Agency.

\bibliographystyle{alpha}
\bibliography{ref}
	
\end{document}